%% file: Globecom_draft.tex
\documentclass[conference]{IEEEtran}
\IEEEoverridecommandlockouts
\usepackage{amsmath,amsfonts}
\usepackage{algorithm}
\usepackage{algpseudocode}
\usepackage{array}
\usepackage[caption=false,font=normalsize,labelfont=sf,textfont=sf]{subfig}
\usepackage{textcomp}
\usepackage{stfloats}
\usepackage{url}
\usepackage{verbatim}
\usepackage{graphicx}
\usepackage{cite}
\def\BibTeX{{\rm B\kern-.05em{\sc i\kern-.025em b}\kern-.08em
    T\kern-.1667em\lower.7ex\hbox{E}\kern-.125emX}}
\usepackage{balance}

\usepackage{makecell}
\usepackage{subfig}  
\usepackage{amssymb}
\usepackage{mathrsfs}
\usepackage{amsthm}
\RequirePackage{epsf}
\usepackage{epstopdf}
\usepackage[printonlyused]{acronym}

\input{acronym}

\newtheorem{theorem}{Theorem}

\newtheorem{lemma}{Lemma}

\newtheorem{remark}{Remark}

\usepackage{siunitx}
\begin{document}

\title{Sensing Safety Analysis for Vehicular Networks with Integrated Sensing and Communication (ISAC)\vspace{-0.2cm}}

\author{\normalsize Tingyu Shui\textsuperscript{1}, Walid Saad\textsuperscript{1}, and Mingzhe Chen\textsuperscript{2} \\ 
\textsuperscript{1}Bradley Department of Electrical and Computer Engineering, Virginia Tech, Alexandria, VA, 22305, USA.\\
\textsuperscript{2}Department of Electrical and Computer Engineering and Frost Institute for Data Science and Computing, \\ University of Miami, Coral Gables, FL, 33146, USA. \\
Emails:\{tygrady, walids\}@vt.edu, mingzhe.chen@miami.edu\vspace{-0.4cm}}
\maketitle

\begin{abstract}
\Ac{ISAC} emerged as a key feature of next-generation 6G wireless systems, allowing them to achieve high data rates and sensing accuracy. While prior research has primarily focused on addressing communication safety in \ac{ISAC} systems, the equally critical issue of sensing safety remains largely ignored. In this paper, a novel threat to the sensing safety of \ac{ISAC} vehicle networks is studied, whereby a malicious \ac{RIS} is deployed to compromise the sensing functionality of a \ac{RSU}. Specifically, a malicious attacker dynamically adjusts the phase shifts of an \ac{RIS} to spoof the sensing outcomes of a \ac{VU}'s echo delay, Doppler shift, and \ac{AoD}. To achieve spoofing on Doppler shift estimation, a time-varying phase shift design on the \ac{RIS} is proposed. Furthermore, the feasible spoofing frequency set with respect to the Doppler shift is analytical derived. Analytical results also demonstrate that the \ac{MLE} of the \ac{AoD} can be significantly misled under spoofed Doppler shift estimation. Simulation results validate our theoretical findings, showing that the \ac{RIS} can induce a spoofed velocity estimation from $0.1$~m/s to $14.9$~m/s for a \ac{VU} with velocity of $10$~m/s, and can cause an \ac{AoD} estimation error of up to $65^\circ$ with only a $5^\circ$ beam misalignment.
\end{abstract}
\acresetall
\section{Introduction}
\Ac{ISAC} is a promising 6G technology for vehicular applications that require both high-quality communication and high-accuracy sensing such as \ac{CAV} and \ac{CP}. \Ac{ISAC} allows a wireless \ac{BS} to simultaneously provide communication and sensing functions through the co-design of the hardware architecture, transmitted waveforms, and signal processing algorithms. However, the joint design of sensing and communications coupled with the broadcast nature of wireless transmission in \ac{ISAC} networks will introduce critical security and safety problems\cite{9755276}.

Aligned with the dual functionalities, \ac{ISAC} networks' security concerns should be addressed from two distinct perspectives: \emph{communication safety} and \emph{sensing safety}. A number of recent works have studied the communication safety issue, in which \ac{ISAC} networks are designed to be resilient to eavesdropping attacks \cite{9199556, 9737364, 10227884,10781436}. Specifically, these studies consider eavesdroppers that act as sensing targets to intercept \ac{ISAC} signals. However, the issue of sensing safety received far less attention. Unlike information leakage, which is the primary concern in communication safety, sensing safety attacks focus on adversarially manipulating and spoofing the sensing outcomes of an \ac{ISAC} system\cite{10839241}. Such attacks on sensing safety may lead to severe consequences. For example, in an \ac{ISAC}-enabled vehicular application, a pedestrian may become undetectable to a \ac{CAV} under a spoofing attack that targets the sensing outcome.

There are only a few preliminary works that looked at this sensing safety problems and its challenges \cite{10856886, 10443321, 10575930, 10634199}. Specifically, these works, from a defender's perspective, aim to conceal targets from being detected by unauthorized \ac{ISAC} stations. Inspired by \ac{ECM} techniques, the authors in \cite{10856886, 10443321, 10575930} leverage a novel \ac{RIS}-aided framework where the sensing echos received at the unauthorized \ac{ISAC} station are suppressed while those received at the authorized \ac{ISAC} station are enhanced. Particularly, the \ac{RIS}, acting as an intelligent \ac{ECM}, is mounted on the target for directional radar stealth. In \cite{10634199}, the \ac{RIS} is further exploited to simultaneously conceal the target and spoof the unauthorized \ac{ISAC} station. By redirecting the detection signal to clutter, the \ac{RIS} generates a deceptive \ac{AoA} sensing outcome on the unauthorized \ac{ISAC} stations. However, the works in \cite{10856886, 10443321, 10575930, 10634199}, which leverage \ac{RIS} to enhance the sensing safety of \ac{ISAC} networks, also raise a critical question: \emph{Can a malicious \ac{RIS} be exploited to compromise the sensing safety of \ac{ISAC} networks instead? If so, how would this malicious \ac{RIS} be designed, and what impact would it have on the sensing outcomes?}

The main contribution of this paper is a novel \ac{RIS}-aided spoofing technique that can be shown to compromise the sensing safety of an \ac{ISAC} vehicular network. In our considered system, an attacker deploys a malicious \ac{RIS} in the coverage of a \ac{RSU}. Different from the directional radar stealth methods proposed in \cite{10856886, 10443321, 10575930, 10634199}, the malicious \ac{RIS} is fixed in location rather than being mounted on the sensing target, such as \acp{VU}, which makes the spoofing much more practical. We particularly propose a time-varying phase shifts design on the \ac{RIS} to spoof the \ac{RSU}'s estimation on the \ac{VU}'s echo delay, Doppler shift, and \ac{AoD}. We then derive the closed-form expressions of the phase shifts design and the feasible spoofing frequencies set with respect to the Doppler shift. Moreover, we analytically show that, the \ac{MLE} of the \ac{AoD} will also be jeopardized under the proposed \ac{RIS} spoofing. To our best knowledge, \textit{this is the first work that designs \ac{RIS} to compromise the sensing safety of \ac{ISAC} network and analyzes its impact on the resulted sensing outcome}. Simulation results demonstrate that, given a $10$~m/s true velocity for the \ac{VU}, the \ac{RIS} can cause a spoofed velocity estimation from $0.1$~m/s to $14.9$~m/s and a spoofed \ac{AoD} estimation of deviation up to $65^\circ$, even with a beam misalignment of only $5^\circ$ for the \ac{RSU}.

The rest of the paper is organized as follows. The system model is presented in Section \ref{Section II}. The design of \ac{RIS} spoofing and the analysis of its impact on sensing safety are performed in Section \ref{Section III}. Section \ref{Section IV} presents the simulation results and Section \ref{Section V} concludes the paper.

\section{System Model}
\label{Section II}
Consider an \ac{ISAC} system supported by a full-duplex \ac{mmWave} \ac{RSU}. The \ac{RSU} is equipped with a \ac{ULA} of $N_t$ transmit antennas and $N_r$ receive antennas to provide sensing and downlink communications to a single-antenna \ac{VU}. By estimating the \ac{VU}'s kinetic states from echo signals, the \ac{RSU} can support critical applications such as beam tracking and \ac{CAV} trajectory planning. A two-dimensional Cartesian coordinate system is used with the \ac{RSU} located at its origin. The goal of the \ac{RSU} is to estimate the coordinates $(x_{\textrm{V}},y_{\textrm{V}})$ and velocity $v$ of the \ac{VU}. However, a malicious \ac{RIS} is deployed to compromise the \ac{RSU}'s sensing functionality by manipulating the echo signals, as shown in Fig. \ref{System_Model}. Following a standard convention in the literature \cite{9171304}, we focus on the case in which the \ac{VU} moves along a straight road parallel to the \ac{RSU}'s \ac{ULA}\footnote{The extension to a non-parallel case is straightforward by rotating the coordinate system.}. 
\subsection{Signal Model}
\vspace{-3pt}
Let $s(t) \in \mathbb{C}$ be the \ac{RSU}'s transmitted \ac{ISAC} symbol with unit power, i.e., $s(t)s^{*}(t) = 1$. Given the precoding vector $\boldsymbol{w} \in \mathbb{C}^{N_t \times 1}$ at the \ac{RSU}, the echo signal reflected by the \ac{VU} can be modeled as follows:
\begin{equation}
\label{echo_V}
    \boldsymbol{y}_{\textrm{E,V}}(t) = \sqrt{P} \gamma_{\textrm{B}} \beta_{\textrm{V}} e^{j 2 \pi \mu_{\textrm{V}} t} \boldsymbol{b}_B(\theta_{\textrm{V}}) \boldsymbol{a}^{H}_B(\theta_{\textrm{V}}) \boldsymbol{w} s(t - \tau_{\textrm{V}}),
\end{equation}
where $P$ is the transmit power of the \ac{RSU}, $\gamma_{\textrm{B}} = \sqrt{N_t N_r}$ is the array gain factor, and $\beta_{\textrm{V}} = \sqrt{\frac{\lambda^2 \kappa_{\textrm{V}}}{64 \pi^3 d_{\textrm{V}}^4}} e^{\frac{j 4 \pi d_{\textrm{V}}}{\lambda}}$ combines the complex path gain of the \ac{RSU}-\ac{VU}-\ac{RSU} link and the \ac{RCS} $\kappa_{\textrm{V}}$ of the \ac{VU}. The distance between the \ac{RSU} and the \ac{VU} is $d_{\textrm{V}}$ and the carrier wavelength is $\lambda$. Meanwhile, the Doppler shift $\mu_{\textrm{V}}$, \ac{AoD} $\theta_{\textrm{V}}$, and double-path echo delay $\tau_{\textrm{V}}$ of the \ac{VU} are included in \eqref{echo_V}, where $\boldsymbol{a}_B(\theta)$ and $\boldsymbol{b}_B(\theta)$ are the steering vector of the transmitting and receiving antennas, respectively, at the \ac{RSU}. By assuming a half-wavelength antenna space, $\boldsymbol{a}_B(\theta)$ and $\boldsymbol{b}_B(\theta)$ will be given by:
\begin{equation}
\label{steering_t} 
    \boldsymbol{a}_B(\theta)= \frac{1}{\sqrt{N_t}} \left[1, e^{-j \pi \cos \theta}, \ldots, e^{-j \pi\left(N_t-1\right) \cos \theta}\right]^T,
\end{equation}    
\vspace{-8pt}
\begin{equation}
\label{steering_r} 
    \boldsymbol{b}_B(\theta)= \frac{1}{\sqrt{N_r}} \left[1, e^{-j \pi \cos \theta}, \ldots, e^{-j \pi\left(N_r-1\right) \cos \theta}\right]^T.
\end{equation}   

\begin{figure}[t]
	\centering
	\includegraphics[scale=0.43]{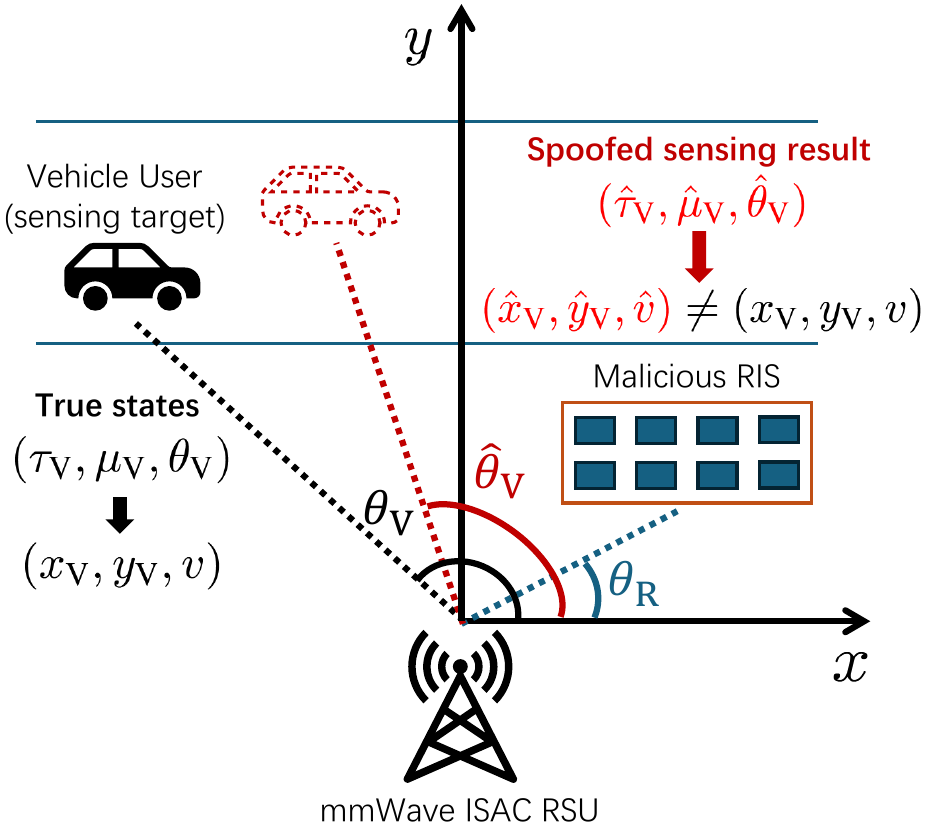}
    \vspace{-5pt}
	\caption{\small{System model of the considered \ac{ISAC} network under the spoofing of a malicious \ac{RIS}.}}
        \vspace{-10pt}
 \label{System_Model}
\end{figure}
The sensing process on the \ac{RSU} should estimate $(x_{\textrm{V}},y_{\textrm{V}})$ and $v$ of the \ac{VU}, which can be derived from $\tau_{\textrm{V}}$, $\mu_{\textrm{V}}$, and $\theta_{\textrm{V}}$. However, in the presence of the malicious \ac{RIS}, the echo received by the \ac{RSU} is a composition of the legitimate echo reflected by the \ac{VU} and the spoofing echo adversarially reflected by the \ac{RIS}. Assume that the reflecting elements of the \ac{RIS} are formed as a \ac{ULA} of $M$ elements, the spoofing echo from the \ac{RIS} can be similarly given by:
\begin{equation}
\label{echo_R}
\begin{aligned}
    \boldsymbol{y}_{\textrm{E,R}}(t) 
    = & \sqrt{P} \gamma_{\textrm{B}} \beta_{\textrm{R}} \boldsymbol{b}_{\textrm{B}}(\theta_{\textrm{R}}) \boldsymbol{a}^{H}_{\textrm{R}}(\theta_{\textrm{B}}) \text{diag}[\boldsymbol{\phi}(t)] \times \\
    & \ \boldsymbol{b}_{\textrm{R}}(\theta_{\textrm{B}}) \boldsymbol{a}^{H}_{\textrm{B}}(\theta_{\textrm{R}}) \boldsymbol{w} s(t - \tau_{\textrm{R}}),
\end{aligned}
\end{equation}
where $\beta_{\textrm{R}} = \sqrt{\frac{\lambda^2 \kappa_{\textrm{R}}}{64 \pi^3 d_{\textrm{R}}^4}} e^{\frac{j 4 \pi d_{\textrm{R}}}{\lambda}}$ combines the complex path gain of the \ac{RSU}-\ac{RIS}-\ac{RSU} link and the \ac{RCS} $\kappa_{\textrm{R}}$ of the \ac{RIS}. According to \cite{9732917}, we have $\kappa_{\textrm{R}} = \frac{4 \pi \eta S^2}{\lambda^2}$ with an \ac{RIS}'s refection efficiency $\eta$, area $S$, and operating wavelength $\lambda$. Similarly, $\boldsymbol{b}_R(\cdot)$ and $\boldsymbol{a}_R(\cdot)$, $\theta_{\textrm{B}}$, and $\tau_{\textrm{R}}$ are the steering vectors, \ac{AoA} (and also \ac{AoD}), and echo delay of the \ac{RIS}. The phase shifts at the \ac{RIS} are given by $\boldsymbol{\phi}(t) = \left[ e^{j \phi_1(t)}, e^{j \phi_2(t)}, \ldots, e^{j \phi_M(t)}\right]$ and unit reflection amplitudes are assumed for simplicity. Note that $\boldsymbol{\phi}(t) $ must be time-varying to spoof the sensing process, as we will explain in Section \ref{Section III}. Thus, the composite echo received by the \ac{RSU} can be given by:
\begin{equation}
\label{joint echo}
    \boldsymbol{y}_{\textrm{E}}(t) = \boldsymbol{y}_{\textrm{E,R}}(t) + \boldsymbol{y}_{\textrm{E,V}}(t) + \boldsymbol{z}_{\textrm{E}}(t),
\end{equation}
where $\boldsymbol{z}_{\textrm{E}}(t) \sim \mathcal{CN}(0, \sigma^2 \boldsymbol{I}_M)$ is the \ac{AWGN} at the \ac{RSU}'s receiving antennas.
\subsection{Sensing Spoofing by Malicious \ac{RIS}}
\vspace{-3pt}
We focus on the sensing process of a short interval $T$, defined as an epoch, during which $(x_{\textrm{V}},y_{\textrm{V}})$ and $v$ are assumed to be constant \cite{9947033}. Within an arbitrary epoch, the \ac{RSU} first determines $\mu_{\textrm{V}}$ and $\theta_{\textrm{V}}$ through a standard matched-filtering technique \cite{9171304}, after which the received signal in \eqref{joint echo} is compensated in both time and frequency domain for further estimating $\theta_{\textrm{V}}$. In particular, the matched-filtering output of $\boldsymbol{y}_{\textrm{E}}(t)$ is defined as follows:
\begin{equation}
\label{match filter}
C(\tau, \mu) \triangleq \left| \int_{0}^{T} \boldsymbol{y}_{\textrm{E}}(t) s^{*}(t - \tau) e^{-j 2 \pi \mu t} dt \right|^2.
\end{equation}
Thus, the estimated echo delay and Doppler shift can be given by $\left( \hat{\tau}_{\textrm{V}}, \ \hat{\mu}_{\textrm{V}} \right) = \arg \max_{\tau, \ \mu} C(\tau, \mu)$. Moreover, we consider the \ac{MLE} of $\theta_{\textrm{V}}$, which can be given by:
\begin{equation}
\label{MLE spoofed}
    \hat{\theta}_{\textrm{V}}=\arg \max _{\theta_{\textrm{V}}} p(\hat{\boldsymbol{y}}_{\textrm{E}} \mid \theta_{\textrm{V}}).
\end{equation}
In \eqref{MLE spoofed}, $\hat{\boldsymbol{y}}_{\textrm{E}} = \int_{0}^{T} \boldsymbol{y}_{\textrm{E}}(t) s^{*}(t - \hat{\tau}_{\textrm{V}}) e^{-j 2 \pi \hat{\mu}_{\textrm{V}} t} dt$ represents the echo compensated by $\hat{\tau}_{\textrm{V}}$ and $\hat{\mu}_{\textrm{V}}$.
\addtolength{\topmargin}{0.03in}
\begin{figure}[t]
	\centering
	\includegraphics[scale=0.6]{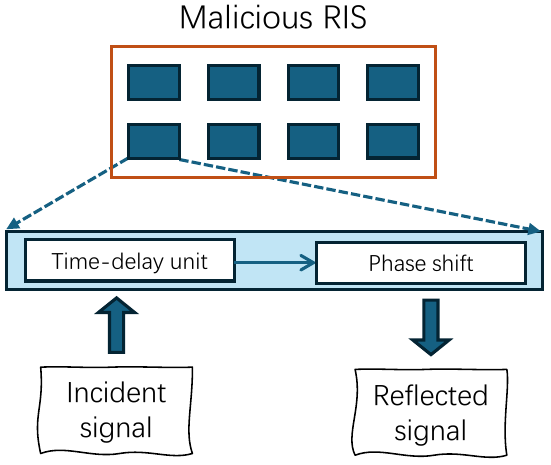}
    \vspace{-5pt}
	\caption{\small{The architecture of the adjustable-delay \ac{RIS}.}}
    \vspace{-10pt}
 \label{time-delay_RIS}
\end{figure}

The goal of the malicious \ac{RIS} is to hinder the \ac{RSU} from obtaining accurate $\hat{\tau}_{\textrm{V}}$, $\hat{\mu}_{\textrm{V}}$, and $\hat{\theta}_{\textrm{V}}$. In other words, the spoofing echo $\boldsymbol{y}_{\textrm{E,R}}(t)$ given in \eqref{echo_R} is expected to cause a deviation in the peak position of $C(\tau, \mu)$ and further distort the accuracy of $\hat{\theta}_{\textrm{V}}$. In typical \ac{ISAC} system of large bandwidth, the delay resolution of $C(\tau, \mu)$ is sufficiently high to resolve the two echoes given in \eqref{echo_V} and \eqref{echo_R}. Thus, the spoofed echo can be mitigated if the echo delay difference satisfies $|\tau_{\textrm{V}} - \tau_{\textrm{R}}| \geq \tau_0$, where $\tau_0$ is the effective main-lobe width of $C(\tau, \mu)$ in the time domain. To successfully conduct the echo spoofing, an adjustable-delay \ac{RIS} \cite{10423078} is implemented in our system, as shown in Fig. \ref{time-delay_RIS}. Here, a time-delay unit is cascaded before the \ac{RIS}'s phase shifter such that it is able to store and retrieve the impinging signals with a designed delay $\Delta t_{\textrm{R}}$. As a result, the echo delay in \eqref{echo_R} will be: $\tau_{\textrm{R}} = \Delta t_{\textrm{R}} + 2 \frac{d_{\textrm{R}}}{c}$. 
\begin{remark}
    The reason for incorporating the adjustable-delay \ac{RIS} is to ensure that $|\tau_{\textrm{V}} - \tau_{\textrm{R}}| < \tau_0$. Thus, $\hat{\tau}_{\textrm{V}}$ is not spoofed. Moreover, hereinafter, we focus on the case $|\tau_{\textrm{V}} - \tau_{\textrm{R}}| = 0$. As a result, a necessary condition for the \ac{RIS} echo spoofing is $0 \leq \Delta t_{\textrm{R}} = 2 \frac{ d_{\textrm{V}} - d_{\textrm{R}} }{c} \leq \Delta_{\textrm{max}}$ with $\Delta_{\textrm{max}}$ being the maximum adjustable-delay.
\end{remark}

\section{Design and Analysis of \ac{RIS} Spoofing}
\label{Section III}
In this section we focus on the design of the \ac{RIS} for spoofing and analyze its impact on the sensing outcome. 
\subsection{Dynamic Phase Shifts Design}
\vspace{-3pt}
First, we can rewrite the matched-filtering output in \eqref{match filter} as:
\begin{equation}
\begin{aligned}
\label{match filtering c}
C(\tau, \mu) = &\Biggl| \sqrt{P} \gamma_{\textrm{B}}  \beta_{\textrm{V}}  \boldsymbol{b}_B(\theta_{\textrm{V}}) \boldsymbol{a}^{H}_B(\theta_{\textrm{V}}) \boldsymbol{w} c_{\textrm{V}}(\tau, \mu ) + \\
& \ \sqrt{P} \gamma_{\textrm{B}}  \beta_{\textrm{R}} \boldsymbol{b}_{\textrm{B}}(\theta_{\textrm{R}}) \boldsymbol{a}^{H}_{\textrm{B}}(\theta_{\textrm{R}}) \boldsymbol{w} c_{\textrm{R}}(\tau, \mu ) + \tilde{\boldsymbol{z}}_{\textrm{E}}(\tau, \mu)   \Biggl |^2, 
\end{aligned}
\end{equation}
where $c_{\textrm{V}}(\tau, \mu ) \triangleq \int_{0}^{T} s(t - \tau_\textrm{V})s^{*}(t - \tau) e^{ - j 2 \pi ( \mu -  \mu_{\textrm{V}}) t} dt $, $c_{\textrm{R}}(\tau, \mu ) \triangleq \int_{0}^{T} s(t - \tau_\textrm{V})s^{*}(t - \tau) \left[ \sum_{m=1}^{M} e^{ - j (2 \pi \mu t - \phi_m(t) )} \right] dt $, and $\tilde{\boldsymbol{z}}_{\textrm{E}}(\tau, \mu) \triangleq \int_{0}^{T} \boldsymbol{z}_{\textrm{E}}(t)s^{*}(t - \tau)e^{-j 2 \pi \mu t} dt $. To spoof the sensing outcome, we assume that the \ac{RIS} can derive $\Delta t_{\textrm{R}} = 2 \frac{ d_{\textrm{V}} - d_{\textrm{R}} }{c}$ by eavesdropping on the uplink communication of the \ac{VU} \cite{9724202}. Thus, $C(\tau, \mu)$ will only exhibit one peak around $\tau_{\textrm{V}}$ in the time domain. Moreover, we assume a perfect echo delay estimation $\hat{\tau}_{\textrm{V}} = \tau_{\textrm{V}}$ due to the high delay resolution in large bandwidth \ac{mmWave} \ac{ISAC} systems. In other words, we only focus on the spoofing on $\mu_\textrm{V}$ and $\theta_\textrm{V}$.

Given \eqref{match filtering c}, we can observe that the impact of \ac{RIS} spoofing stems from the term $\sqrt{P} \gamma_{\textrm{B}} \beta_{\textrm{R}} \boldsymbol{b}_{\textrm{B}}(\theta_{\textrm{R}}) \boldsymbol{a}^{H}_{\textrm{B}}(\theta_{\textrm{R}}) \boldsymbol{w} c_{\textrm{R}}(\tau, \mu )$. Assuming $\hat{\tau}_{\textrm{V}} = \tau_{\textrm{V}}$, we have $c_{\textrm{V}}(\hat{\tau}_{\textrm{V}}, \mu ) = \int_{0}^{T} e^{ - j 2 \pi ( \mu -  \mu_{\textrm{V}}) t} dt $ and $c_{\textrm{R}}(\hat{\tau}_{\textrm{V}}, \mu ) = \sum_{m=1}^{M}  \int_{0}^{T} e^{ - j (2 \pi \mu t - \phi_m(t)) )}dt $. We can find that the peak position of $c_{\textrm{V}}(\hat{\tau}_{\textrm{V}}, \mu )$ in the frequency domain occurs at $\mu = \mu_{\textrm{V}}$. Thus, a natural analogy is to set $\phi_m(t) = 2 \pi \mu_m t$ for artificially creating $M$ peak positions at frequencies $\mu_{1}, \ldots, \mu_{M}$. In practice, due to the hardware constraints of the \ac{RIS}, the phase shift $\phi_m(t)$ on the $m$-th element can be only given in a discrete form as:
\begin{equation}
\label{phase shift}
 \phi_{m}(t) = (2 \pi \mu_m \lceil \frac{t}{\Delta T} \rceil \Delta T) \mod{2 \pi},
\end{equation}
where $\Delta T$ represents the shortest time interval over which $\phi_{m}(t)$ can vary. Meanwhile, the modulo operation is applied because the phase shifts of the \ac{RIS} are typically confined to the range $0$ to $2\pi$. Given the design in \eqref{phase shift}, we focus on one fundamental spoofing case in which the $M$ spoofing frequencies are set equal, i.e, $\mu_1 = \mu_2 = \ldots = \mu_M = \tilde{\mu}$ such that $\hat{\mu}_{\textrm{V}} \triangleq \arg \max_{\mu} C(\hat{\tau}_{\textrm{V}}, \mu)$ is spoofed as $\tilde{\mu}$. To this end, the \ac{RIS} should select $\tilde{\mu}$ such that the peak of $C(\hat{\tau}_{\textrm{V}}, \mu)$ occurs at $\tilde{\mu}$ rather than $\mu_{\textrm{V}}$, i.e., $C(\hat{\tau}_{\textrm{V}}, \tilde{\mu}) \geq  C(\hat{\tau}_{\textrm{V}}, \mu_{\textrm{V}})$. 

\subsection{Impact on Doppler Shift Estimation}
\vspace{-3pt}
We begin by deriving the feasible spoofing frequency set $\mathcal{A} \triangleq \{ \tilde{\mu} \mid C(\hat{\tau}_{\textrm{V}}, \tilde{\mu}) \geq  C(\hat{\tau}_{\textrm{V}}, \mu_{\textrm{V}}) \}$. Since the matched-filtering output of noise is negligible compared to the spoofing echo, we can ignore $\tilde{\boldsymbol{z}}_{\textrm{E}}(\hat{\tau}_{\textrm{V}}, \mu)$ and rewrite $C(\hat{\tau}_{\textrm{V}}, \mu)$ as:
\begin{equation}
\label{match filter square}
\begin{aligned}
& C(\hat{\tau}_{\textrm{V}}, \mu) \\
= & \frac{P\gamma_{\textrm{B}}^2}{N_r} \sum_{n_r=1}^{N_r}  \Biggl|  \beta_{\textrm{V}} g_{n_r}(\theta_0, \theta_{\textrm{V}})  c_{\textrm{V}}(\hat{\tau}_{\textrm{V}}, \mu ) + \\
& \ \beta_{\textrm{R}} g_{n_r}(\theta_0, \theta_{\textrm{R}}) M c_{\textrm{R}}(\hat{\tau}_{\textrm{V}}, \mu ) \Biggr|^2 \\
= & \frac{PT^2\gamma_{\textrm{B}}^2}{N_r} \sum_{n_r=1}^{N_r}  \Biggl|  \beta_{\textrm{V}} g_{n_r}(\theta_0, \theta_{\textrm{V}}) e^{-j\pi(\mu-\mu_{\textrm{V}})T}  \operatorname{sinc}(T (\mu - \mu_{\textrm{V}})) + \\
& \ \beta_{\textrm{R}} g_{n_r}(\theta_0, \theta_{\textrm{R}}) M e^{j \pi \tilde{\mu} \Delta T} e^{- j \pi (\mu - \tilde{\mu})T} \operatorname{sinc}(\mu \Delta T) \times \\
& \ f(K,\pi \Delta T (\mu - \tilde{\mu}))\Biggr|^2,
\end{aligned}
\end{equation}
where we adopt a precoding vector $\boldsymbol{w} = \boldsymbol{a}_B(\theta_{0})$ steered towards $\theta_0$, which may be selected from a predefined codebook \cite{10740590}. Moreover, in \eqref{match filter square}, we define $f(K,x) = \frac{\sin{(Kx)}}{K \sin{x}}$, $\operatorname{sinc}(x) = \frac{\sin(\pi x)}{\pi x}$, and $g_{n_r}(\theta_1, \theta_2) = e^{-\frac{j \pi}{2} [ \cos{\theta_1(N_t -1)} - \cos{\theta_2} (N_t + 1 - 2n_r) ] } f(N_t, \frac{\pi}{2} (\cos{\theta_1} - \cos{\theta_{2}}) )$. To derive the expression of $\mathcal{A}$, we first present the following lemma pertaining to a set of infeasible spoofing frequencies.
\addtolength{\topmargin}{0.03in}
\begin{lemma}
\label{lemma 1}
    If the number of \ac{RIS} reflecting elements satisfies $M \gg \sqrt{ \frac{\kappa_{\textrm{V}}}{4 \pi \eta}} \frac{\lambda}{S}  \left| \frac{f(N_t, \frac{\pi}{2} (\cos{\theta_0} - \cos{\theta_{\textrm{V}}}) )}{f(N_t, \frac{\pi}{2} (\cos{\theta_0} - \cos{\theta_{\textrm{R}}}) )} \right|$, an infeasible spoofing frequency set can be given by $ \mathcal{A}_{\varnothing} = \left\{ \mu \mid \mu = \mu_{\textrm{V}} + \frac{n}{\Delta T}, n \in \mathcal{Z} \right\}$, i.e., $\mathcal{A}_{\varnothing} \cap \mathcal{A} = \varnothing$.
\end{lemma}
\begin{proof}
To prove Lemma \ref{lemma 1}, we can directly show that the inequality $C(\hat{\tau}_{\textrm{V}}, \mu_{\textrm{V}}) \geq C(\hat{\tau}_{\textrm{V}}, \tilde{\mu})$ holds for $\tilde{\mu} \in \mathcal{A}_\varnothing$ as follows:
\begin{equation}
\begin{aligned}
    & C(\hat{\tau}_{\textrm{V}}, \mu_{\textrm{V}}) \\
    = & \frac{PT^2\gamma_{\textrm{B}}^2}{N_r} \sum_{n_r=1}^{N_r}  \Biggl| \beta_{\textrm{V}} g_{n_r}(\theta_0, \theta_{\textrm{V}}) + f(K,- \pi n)e^{ j \pi \frac{n}{\Delta T}T} \times\\
    & \ M \beta_{\textrm{R}} g_{n_r}(\theta_0, \theta_{\textrm{R}}) e^{j \pi \mu_{\varnothing} \Delta T} \operatorname{sinc}(\mu_{\textrm{V}} \Delta T) \Biggr|^2 \\
    \overset{(a)}{\approx} & \frac{PT^2\gamma_{\textrm{B}}^2}{N_r} \sum_{n_r=1}^{N_r} \Biggl| M \beta_{\textrm{R}} g_{n_r}(\theta_0, \theta_{\textrm{R}}) e^{j \pi \mu_{\varnothing} \Delta T} \operatorname{sinc}(\mu_{\textrm{V}} \Delta T) \Biggr|^2 \\
    \overset{(b)}{>} & \frac{PT^2\gamma_{\textrm{B}}^2}{N_r} \sum_{n_r=1}^{N_r}  \Biggl| M \beta_{\textrm{R}} g_{n_r}(\theta_0, \theta_{\textrm{R}}) e^{j \pi \mu_\varnothing \Delta T} \operatorname{sinc}(\mu_{\textrm{V}}\Delta T + n) \Biggr|^2 \\
    = & C(\hat{\tau}_{\textrm{V}}, \mu_{\varnothing}).
\end{aligned}
\end{equation}
Approximation (a) holds for $\forall n_r \in \left\{1, \ldots, N_r\right\}$, when $\left| f(K,- \pi n)e^{ j \pi \frac{n}{\Delta T}T} M \beta_{\textrm{R}} g_{n_r}(\theta_0, \theta_{\textrm{R}}) e^{j \pi \mu_{\varnothing} \Delta T} \operatorname{sinc}(\mu_{\textrm{V}} \Delta T) \right| \gg \left| \beta_{\textrm{V}} g_{n_r}(\theta_0, \theta_{\textrm{V}}) \right| $, which is equivalent to the condition $M \gg \sqrt{ \frac{\kappa_{\textrm{V}}}{4 \pi \eta}} \frac{\lambda}{S}  \left| \frac{f(N_t, \frac{\pi}{2} (\cos{\theta_0} - \cos{\theta_{\textrm{V}}}) )}{f(N_t, \frac{\pi}{2} (\cos{\theta_0} - \cos{\theta_{\textrm{R}}}) )} \right|$. Inequality (b) holds because $ \left| \operatorname{sinc}(\mu_{\textrm{V}}\Delta T ) \right| > \left| \operatorname{sinc}(\mu_{\textrm{V}}\Delta T + n) \right|, \forall n \in \mathcal{Z}$. 
\end{proof}

Given Lemma 1, if we only consider $\tilde{\mu} \in \mathcal{A}$, i.e., $\tilde{\mu} \notin \mathcal{A}_\varnothing$, $C(\hat{\tau}_{\textrm{V}}, \mu)$ in \eqref{match filter square} can be approximated by:
\begin{equation}
\label{match filter approx}
\begin{aligned}
C(\hat{\tau}_{\textrm{V}}, \mu) & \approx PT^2\gamma_{\textrm{B}}^2 \Biggl[ C_\textrm{V} \operatorname{sinc}^2(T (\mu - \mu_{\textrm{V}})) + \\
& \quad M^2 C_\textrm{R} \operatorname{sinc}^2(\mu \Delta T)  f^2(K,\pi \Delta T (\mu - \tilde{\mu})) \Biggr],
\end{aligned}
\end{equation}
where the cross-product terms in the quadratic expansion are ignored since $\ f(K,\pi \Delta T (\mu - \tilde{\mu}))\operatorname{sinc}(T (\mu - \mu_{\textrm{V}})) \approx 0$ when $\tilde{\mu} \notin \mathcal{A}_\varnothing$. In \eqref{match filter approx}, we further define $C_\textrm{V} \triangleq \beta_{\textrm{V}}^2 f^2(N_t, \frac{\pi}{2} (\cos{\theta_0} - \cos{\theta_{\textrm{V}}}) )$ and $C_\textrm{R} \triangleq \beta_{\textrm{R}}^2 f^2(N_t, \frac{\pi}{2} (\cos{\theta_0} - \cos{\theta_{\textrm{R}}}) )$. Another observation in \eqref{match filter approx} is that the impact of the spoofing frequency $\tilde{\mu}$ is periodic, as shown in the periodic spoofing term $M^2 C_\textrm{R} \operatorname{sinc}^2(\mu \Delta T)  f^2(K,\pi \Delta T (\mu - \tilde{\mu}))$. Specifically, the spoofing frequency $\tilde{\mu}$ will lead to multiple peaks at $\mu = \tilde{\mu} + \frac{n}{\Delta T}, n \in \mathcal{Z}$. Therefore, in order to derive $\mathcal{A}$, we only consider the highest peak defined in Lemma \ref{lemma 2}.
\begin{lemma}
\label{lemma 2}
    If the condition on $M$ in Lemma \ref{lemma 1} is satisfied, the highest peak of the spoofing term in \eqref{match filter approx} will be given by $ \Delta \tilde{\mu} \triangleq \tilde{\mu} \mod{ \frac{1}{\Delta T}}$.
\end{lemma}
\begin{proof}
Since the periodic function $f^2(K,\pi \Delta T (\mu - \tilde{\mu}))$ has sharp peaks at $\mu = \tilde{\mu} + \frac{n}{\Delta T}, n \in \mathcal{Z}$, the product $\operatorname{sinc}^2(\mu \Delta T)  f^2(K,\pi \Delta T (\mu - \tilde{\mu}))$ will also have peaks at these frequencies. Moreover, the highest peak among $\mu = \tilde{\mu} + \frac{n}{\Delta T}, n \in \mathcal{Z}$ depends on the magnitude of $\operatorname{sinc}^2(\mu \Delta T)$. Given $\mu = \tilde{\mu} + \frac{n}{\Delta T}, n \in \mathcal{Z}$, we further observe that $\operatorname{sinc}^2(\mu \Delta T) = \operatorname{sinc}^2(\tilde{\mu} \Delta T + n)$. Thus, we can find that the highest magnitude of $\operatorname{sinc}^2(\mu \Delta T)$ locates at $\Delta \tilde{\mu} \triangleq \tilde{\mu} \mod{ \frac{1}{\Delta T}}$, i.e., the positive frequency closest to $0$.
\end{proof}
\noindent
Given the result in Lemma \ref{lemma 2}, the possible spoofing frequency is restricted in $\left(0, \frac{1}{\Delta T}\right]$. In other words, an \ac{RIS} capable changing its phase shift more frequently, i.e., operating with a smaller $\Delta T$, will have a broader range of spoofing frequencies $\Delta \tilde{\mu}$. Moreover, the feasible spoofing frequency set should be redefined as $\mathcal{A} \triangleq \left\{ \Delta \tilde{\mu} \mid C(\hat{\tau}_{\textrm{V}}, \Delta\tilde{\mu}) \geq  C(\hat{\tau}_{\textrm{V}}, \mu_{\textrm{V}})\right\}$, which is derived next.
\begin{figure*}[t]
\vspace{-12pt}
    \begin{equation}
    \label{feasible set}
    \mathcal{A} = \left\{ \Delta \tilde{\mu}  \mid M^2 C_\textrm{R} \left[ \operatorname{sinc}^2(\Delta \tilde{\mu} \Delta T) - \operatorname{sinc}^2(\mu_{\textrm{V}} \Delta T)f^2(K,\pi \Delta T (\mu_{\textrm{V}} - \Delta \tilde{\mu})) \right] - C_\textrm{V} \left[ 1 - \operatorname{sinc}^2(T (\mu_{\textrm{V}} - \Delta \tilde{\mu})) \right] \geq 0  \right\}.
\end{equation}
\vspace{-25pt}
\end{figure*}
\begin{theorem}
\label{theorem 1}
Under the condition on $M$ in Lemma \ref{lemma 1}, the feasible spoofing frequency set is given in \eqref{feasible set}.
\end{theorem}
\begin{proof}
The proof can be directly completed by finding $\Delta \tilde{\mu}$ that satisfies $C(\hat{\tau}_{\textrm{V}}, \Delta \tilde{\mu}) \geq C(\hat{\tau}_{\textrm{V}}, \mu_{\textrm{V}})$ based on \eqref{match filter approx}.
\end{proof}

Theorem \ref{theorem 1} indicates that, if the \ac{RIS} selects a spoofing frequency $\Delta \tilde{\mu} \in \mathcal{A}$, the estimated Doppler shift will be spoofed as $\hat{\mu}_{\textrm{V}} = \Delta \tilde{\mu}$. According to the derived result in \eqref{feasible set}, the range of feasible spoofing frequencies expands as more reflecting elements are deployed at the \ac{RIS}. In addition, the beam orientation $\theta_0$ also influences the set $\mathcal{A}$. Generally, when $\theta_0$ is closer to the \ac{VU}'s \ac{AoD} $\theta_{\textrm{V}}$, the spoofing set $\mathcal{A}$ becomes narrower, whereas a beam direction closer to the \ac{RIS}'s \ac{AoD} $\theta_{\textrm{R}}$ leads to a wider set. This phenomenon poses a critical challenge when \ac{ISAC} is leveraged for beam tracking. For instance, when a \ac{VU} moves through the coverage of the \ac{RSU}, there inevitably exists moments when the beam is steered closer to the \ac{RIS}. At such instances, the \ac{RIS} can select a spoofing frequency $\Delta \tilde{\mu}$ that results significant beam misalignment in the subsequent tracking step. More critically, this misalignment may accumulate over time, leading the \ac{RSU} to lose track of the \ac{VU} eventually.
\subsection{Impact on \ac{AoD} \ac{MLE}} 
\vspace{-3pt}
Given the estimated $\hat{\tau}_{\textrm{V}}$ and $\hat{\mu}_{\textrm{V}}$, the \ac{RSU} will continue to estimate $\theta_{\textrm{V}}$. First, the compensated and normalized echo is given by:
\begin{equation}
\label{compensate signal}
\hat{\boldsymbol{y}}_{\textrm{E}} \triangleq \frac{1}{\sqrt{P} \gamma_{\textrm{B}}} \int_{0}^{T} \boldsymbol{y}_{\textrm{E}}(t) s^{*}(t - \hat{\tau}) e^{-j 2 \pi \Delta \tilde{\mu} t} dt.
\end{equation}
We next define a perfect \ac{MLE} $\tilde{\theta}_{\textrm{V}}$ obtained without spoofing and examine the deviation between the spoofed \ac{MLE} $\hat{\theta}_\textrm{V}$ and $\tilde{\theta}_{\textrm{V}}$. First, we rewrite \eqref{compensate signal} as:
\begin{equation}
\begin{aligned}
\hat{\boldsymbol{y}}_{\textrm{E}}
= & \ T \beta_{\textrm{V}}  \boldsymbol{b}_B(\theta_{\textrm{V}}) h(\theta_{\textrm{V}}, \theta_{0}) e^{-j\pi(\Delta \tilde{\mu}-\mu_{\textrm{V}})T}  \operatorname{sinc}(T (\Delta \tilde{\mu} - \mu_{\textrm{V}})) \\
& \ + M T \beta_{\textrm{R}} \boldsymbol{b}_{\textrm{B}}(\theta_{\textrm{R}}) h(\theta_{\textrm{R}}, \theta_{0}) e^{j \pi \Delta \tilde{\mu} \Delta T} \operatorname{sinc}(\Delta \tilde{\mu} \Delta T) + \hat{\boldsymbol{z}}_{\textrm{E}},
\end{aligned}
\end{equation}
where $\hat{\boldsymbol{z}}_{\textrm{E}} = \frac{\tilde{\boldsymbol{z}}_{\textrm{E}}(\hat{\tau}, \Delta \tilde{\mu}) }{\sqrt{P}\gamma_{\textrm{B}}} \sim \mathcal{C N}\left(\boldsymbol{0}_{N_r}, \frac{\sigma^2T}{P \gamma_{\textrm{B}}^2} \boldsymbol{I}_{N_r}\right)$ and $h(\theta_1, \theta_2) = e^{-\frac{j \pi(N_t -1)}{2} [ \cos{\theta_1} - \cos{\theta_2} ] } f(N_t, \frac{\pi}{2} (\cos{\theta_1} - \cos{\theta_{2}}) )$. Similar to \eqref{MLE spoofed}, we define the perfect \ac{MLE} without spoofing as follows:
\begin{equation}
\label{MLE perfect}
    \tilde{\theta}_{\textrm{V}}=\arg \max _{\theta_{\textrm{V}}} p( \hat{\boldsymbol{y}}\mid \theta_{\textrm{V}}),
\end{equation}
where $\hat{\boldsymbol{y}} = T \beta_{\textrm{V}}   \boldsymbol{b}_B(\theta_{\textrm{V}}) h(\theta_{\textrm{V}}, \theta_{0}) + \hat{\boldsymbol{z}}_{\textrm{E}}$. Here, $\hat{\boldsymbol{y}}$ represents the normalized received signal compensated by perfect estimation $\hat{\tau}_{\textrm{V}} = \tau_{\textrm{V}}$ and $\hat{\mu}_{\textrm{V}} = \mu_{\textrm{V}}$, without the \ac{RIS} spoof. Thus, $\tilde{\theta}_\textrm{V}$ in \eqref{MLE perfect} is the best \ac{MLE} on $\theta_{\textrm{V}}$ we can obtain. 

Next, we derive the deviation between the perfect \ac{MLE} $\tilde{\theta}_{\textrm{V}}$ and the spoofed \ac{MLE} $\hat{\theta}_{\textrm{V}}$ in the following theorem.
\begin{theorem}
\label{theorem 2}
Assume the perfect delay estimation $\hat{\tau}_\textrm{V} = \tau_{\textrm{V}}$ and spoofed Doppler shift estimation $\hat{\mu}_{\textrm{V}} = \Delta \tilde{\mu}$ are obtained under \ac{RIS} spoofing. The spoofed \ac{MLE} $\hat{\theta}_{\textrm{V}}$ can be given by:
\begin{equation}
\label{MLE spoofed analytical in theorem}
\begin{aligned}
\hat{\theta}_{\textrm{V}} = & \arg \min _{\theta_{\textrm{V}}} \Biggl [  \left\| \hat{\boldsymbol{y}} - T \beta_{\textrm{V}} \boldsymbol{b}_B(\theta_{\textrm{V}}) h(\theta_{\textrm{V}}, \theta_{0})\right\|^2 +  \\
& \ 2 T \beta_{\textrm{V}} \sum_{n=1}^{N_r} \Re{\left\{ \Delta \hat{y}_n g_{n}^{*}(\theta_{\textrm{V}}, \theta_0) \right\}} \Biggr ],
\end{aligned}
\end{equation}
where $\Delta \hat{\boldsymbol{y}} \triangleq \hat{\boldsymbol{y}}_{\textrm{E}} - \hat{\boldsymbol{y}} = \left[ \Delta  \hat{y}_1, \ldots,  \Delta  \hat{y}_{N_r} \right]^{T}$. Moreover, the perfect \ac{MLE} $\hat{\theta}_{\textrm{V}}$ defined in \eqref{MLE perfect} can be similarly given by: 
\begin{equation}
\label{MLE perfect analytical}
\begin{aligned}
\tilde{\theta}_{\textrm{V}} = \arg \min _{\theta_{\textrm{V}}} \left\|\hat{\boldsymbol{y}} -  T\beta_{\textrm{V}} \boldsymbol{b}_B(\theta_{\textrm{V}}) h(\theta_{\textrm{V}}, \theta_{0})\right\|^2.
\end{aligned}
\end{equation}
\end{theorem}
\begin{proof}
We first derive $\tilde{\theta}_{\textrm{V}}$ defined in \eqref{MLE perfect}. Since $\hat{\boldsymbol{y}} = T \beta_{\textrm{V}}   \boldsymbol{b}_B(\theta_{\textrm{V}}) h(\theta_{\textrm{V}}, \theta_{0}) + \hat{\boldsymbol{z}}_{\textrm{E}}$ and $\hat{\boldsymbol{z}}_{\textrm{E}} \sim \mathcal{C N}\left(\boldsymbol{0}_{N_r}, \frac{\sigma^2}{P \gamma_{\textrm{B}}^2} \boldsymbol{I}_{N_r}\right)$, \eqref{MLE perfect analytical} can be directly obtained based on the \ac{PDF} of multivariate normal distribution. For the spoofed \ac{MLE}, since the \ac{RSU} is unaware of the presence of the malicious \ac{RIS}, $\hat{\theta}_{\textrm{V}}$ is obtained by directly replacing $\hat{\boldsymbol{y}}$ in \eqref{MLE perfect analytical} with the superimposed echo $\hat{\boldsymbol{y}}_{\textrm{E}}$ as:
\begin{equation}
\label{MLE spoofed analytical}
\begin{aligned}
\hat{\theta}_{\textrm{V}} = \arg \min _{\theta_{\textrm{V}}} \left\|\hat{\boldsymbol{y}}_{\textrm{E}} - T \beta_{\textrm{V}} \boldsymbol{b}_B(\theta_{\textrm{V}}) h(\theta_{\textrm{V}}, \theta_{0})\right\|^2.
\end{aligned}
\end{equation}
We can further derive \eqref{MLE spoofed analytical} as following
\begin{equation}
\begin{aligned}
\hat{\theta}_{\textrm{V}}
= & \arg \min _{\theta_{\textrm{V}}} \left\|\hat{\boldsymbol{y}}_{\textrm{E}} - \hat{\boldsymbol{y}} + \hat{\boldsymbol{y}} - T \beta_{\textrm{V}} \boldsymbol{b}_B(\theta_{\textrm{V}}) h(\theta_{\textrm{V}}, \theta_{0})\right\|^2 \\
= & \arg \min _{\theta_{\textrm{V}}} \left\| \hat{\boldsymbol{y}} - T \beta_{\textrm{V}} \boldsymbol{b}_B(\theta_{\textrm{V}}) h(\theta_{\textrm{V}}, \theta_{0})\right\|^2 - \\
& \ 2 \Re{\left\{ (\hat{\boldsymbol{y}}_{\textrm{E}} - \hat{\boldsymbol{y}})^{H} (\hat{\boldsymbol{y}} - T \beta_{\textrm{V}} \boldsymbol{b}_B(\theta_{\textrm{V}}) h(\theta_{\textrm{V}}, \theta_{0}))  \right\}} \\
= & \arg \min _{\theta_{\textrm{V}}} \Biggl [  \left\| \hat{\boldsymbol{y}} - T \beta_{\textrm{V}} \boldsymbol{b}_B(\theta_{\textrm{V}}) h(\theta_{\textrm{V}}, \theta_{0})\right\|^2 +  \\
& \ 2 T \beta_{\textrm{V}} \sum_{n=1}^{N_r} \Re{\left\{ \Delta \hat{y}_n g_{n}^{*}(\theta_{\textrm{V}}, \theta_0) \right\}} \Biggr ],
\end{aligned}
\end{equation}
which completes the proof.
\end{proof}
\noindent

From Theorem \ref{theorem 2}, we can observe that the spoofed \ac{MLE} deviates from the perfect \ac{MLE} because of the term $2 T \beta_{\textrm{V}} \sum_{n=1}^{N_r} \Re{\left\{ \Delta \hat{y}_n g_{n}^{*}(\theta_{\textrm{V}}, \theta_0) \right\}}$, which is introduced by the malicious \ac{RIS}. Although it is challenging to derive a closed-form expression for $\left| \tilde{\theta}_{\textrm{V}} - \hat{\theta}_{\textrm{V}}\right|$, we can still find that a spoofed \ac{MLE} $ \hat{\theta}_{\textrm{V}} \neq \tilde{\theta}_{\textrm{V}}$ will be obtained if the minimum of $\left\| \hat{\boldsymbol{y}} - T \beta_{\textrm{V}} \boldsymbol{b}_B(\theta_{\textrm{V}}) h(\theta_{\textrm{V}}, \theta_{0})\right\|^2$ shifts after adding the term $2 T \beta_{\textrm{V}} \sum_{n=1}^{N_r} \Re{\left\{ \Delta \hat{y}_n g_{n}^{*}(\theta_{\textrm{V}}, \theta_0) \right\}}$.

\section{Simulation Results and Analysis}
\label{Section IV}
For our simulations, we consider a two dimensional Cartesian coordinate system with the \ac{RSU} located at its origin. The \ac{AoD}s of the \ac{VU} and \ac{RIS} are $\theta_{\textrm{V}} = 135^{\circ}$ and $\theta_{\textrm{R}} = 90^{\circ}$. The coordinates of the \ac{VU} and \ac{RIS} are $(x_{\textrm{V}}, y_{\textrm{V}}) = (40 \times \cos{\theta_{\textrm{V}}}, 40 \times \sin{\theta_{\textrm{V}}})$ and $(x_{\textrm{R}}, y_{\textrm{R}}) = (30 \times \cos{\theta_{\textrm{R}}}, 30 \times \sin{\theta_{\textrm{R}}})$ both in meters. The other parameters are set as follows unless specified otherwise later: $P = 30$~dBm, $\sigma^2 = -130$~dBm, $f_c = 28$~GHz, $N_t = N_r = 32$, $M = 32$, $\kappa_{\textrm{V}} = 7$~dBsm, $\eta = 0.8$, $S = 50~\text{cm} \times 10~\text{cm}$, $T = 10$~ms, and $\Delta T = 1$~ms. The \ac{VU} is assumed to move along the positive x-axis at a speed of $v = 10$~m/s, and, thus, the Doppler shift can be derived as $\mu_{\textrm{V}} = v f_c \cos{(\pi - \theta_{\textrm{V}})}$.
 
\begin{figure}[t]
	\centering	\includegraphics[scale=0.48]{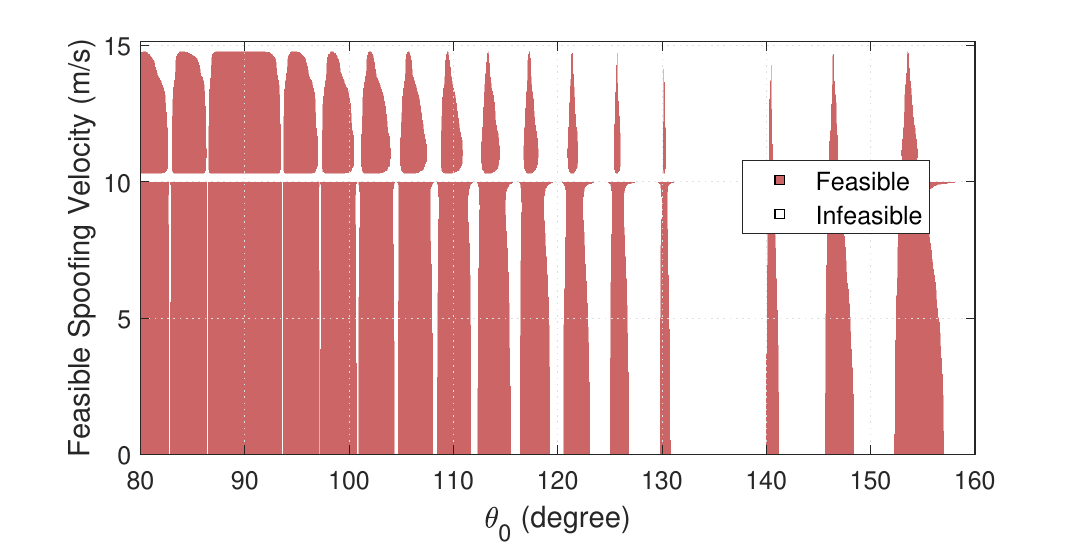}
	\caption{\small{Feasible spoofing velocity versus beam steering direction.}}\
    \vspace{-10pt}
    \label{feasible_spoofing_velocity}
\end{figure}
Fig. \ref{feasible_spoofing_velocity} shows the impact of \ac{RIS} spoofing on the estimated velocity of the \ac{VU}, which can be derived by the spoofed Doppler shift. In particular, the different feasible spoofing velocity sets, derived from \eqref{feasible set}, versus \ac{RSU}'s beam steering direction $\theta_0$ are illustrated. From Fig. \ref{feasible_spoofing_velocity}, we can observe that, generally, a beam steered closer to \ac{VU}'s \ac{AoD} $\theta_{\textrm{V}} = 135^{\circ}$ can eliminate the impact of \ac{RIS} spoofing. For instance, when $\theta_0 \in (131^{\circ}, 140^{\circ})$, we have $\mathcal{A} = \varnothing$, indicating that no successful spoofing can be conducted. However, as $\theta_0$ is in the vicinity of $\theta_{\textrm{R}} = 90^{\circ}$, i.e., steered towards the \ac{RIS}, the \ac{RSU} may obtain a spoofed velocity ranging from $0.1$~m/s to $14.9$~m/s, introducing a significant deviation of $-9.9$~m/s to $4.9$~m/s with respect to the true value $v = 10$~m/s.

\begin{figure}[t]
	\centering	\includegraphics[scale=0.52]{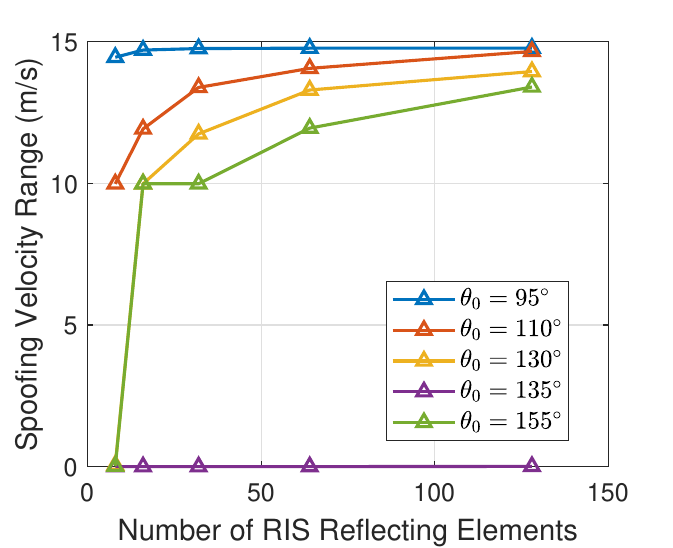}
	\caption{\small{Spoofing velocity range versus number of \ac{RIS} reflecting elements.}}
    \label{spoofing_velocity_range}
\end{figure}
Fig. \ref{spoofing_velocity_range} shows the range of feasible spoofing velocity, i.e., the difference between maximal and minimal spoofing velocities, versus the number of \ac{RIS} reflecting elements. We can see that the range of feasible spoofing velocity expands as the number of reflecting elements increases, except when the \ac{RSU}'s beam is set to be perfectly steered towards the \ac{VU} with $\theta_0 = 135^{\circ}$. Moreover, we can observe that the range is large for $\theta_0 = 90^{\circ}$ and $\theta_0 = 110^{\circ}$ since these $\theta_0$ are close to the \ac{RIS}. However, the range is only $1$~m/s lower for $\theta_0 = 130^{\circ}$ with only $5^{\circ}$ misalignment with respect to $\theta_{\textrm{V}}$. Moreover, even if the beam direction is far to both \ac{VU} and \ac{RIS}, the spoofing can be still conducted as long as the \ac{RIS} is equipped with enough reflecting elements, as shown in the case $\theta_0 = 155^{\circ}$.

\begin{figure}[t]
	\centering	\includegraphics[scale=0.54]{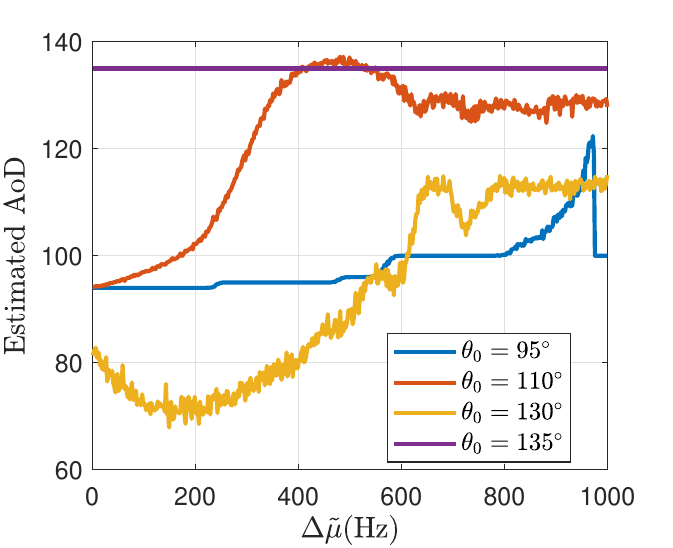}
	\caption{\small{Impact of Doppler shift spoofing on the resulted \ac{MLE} of \ac{AoD}.}}
    \label{spoofed_AoD}
\end{figure}
In Fig. \ref{spoofed_AoD}, we show the spoofed \ac{MLE} for \ac{AoD} as a function of the spoofing frequency on the Doppler shift estimation. The estimated \ac{AoD} are averaged over \num{2000} trials of independent noise. First, we can observe that the \ac{MLE} for \ac{AoD} will not be affected under a beam perfectly steered to the \ac{VU} when $\theta_0 = 135^{\circ}$. However, the \ac{RIS} can always find a spoofing frequency that causes a severe deviation of the spoofed \ac{MLE} under beam misalignment. For instance, when $\theta_0 = 130^{\circ}$ is only $5^{\circ}$ deviating from $\theta_{\textrm{V}}$, its \ac{MLE} can be spoofed as $70^{\circ}$ with $\Delta \tilde{\mu} = 180$~Hz, which leads to an estimation error of $65^{\circ}$. 

\vspace{-3pt}
\section{Conclusion}
\label{Section V}
\vspace{-3pt}
In this paper, we have investigated a novel \ac{RIS}-aided spoofing strategy that compromises the sensing safety of the \ac{ISAC} system. In particular, we have proposed a time-varying phase shift design at the \ac{RIS} to spoof the \ac{RSU}'s sensing outcome of the \ac{VU}'s Doppler shift and \ac{AoD}. Moreover, we have analytically derived the feasible spoofing frequency set with respect to the Doppler shift and the deviation of the spoofed \ac{MLE} for \ac{AoD}. The derived expression has indicated that the number of reflecting elements at the \ac{RIS} and its \ac{AoD} with respect to the \ac{RSU} significantly affect the spoofed sensing outcome. Simulation results demonstrate that, even with a beam misalignment of  only $5^\circ$ with respect to the \ac{VU}, the \ac{RIS}-aided spoofing can lead to a relative velocity estimation error ranging from $-9.9$~m/s to $4$~m/s and an \ac{AoD} estimation error of up to $65^\circ$. 

\bibliographystyle{IEEEtran}
\bibliography{bibliography}
\end{document}

%% file: acronym.tex
\acrodef{C-V2X}{cellular vehicle-to-everything}
\acrodef{V2I}{vehicle-to-infrastructure}
\acrodef{V2V}{vehicle-to-vehicle}
\acrodef{V2P}{vehicle-to-pedestrian}
\acrodef{RB}{resource block}
\acrodef{NR}{new radio}
\acrodef{QoS}{quality-of-service}
\acrodef{SINR}{signal-to-interference-plus-noise ratio}
\acrodef{CSI}{channel state information}
\acrodef{BS}{base station}
\acrodef{TxV}{transmitting vehicle}
\acrodef{RxV}{receiving vehicle}
\acrodef{RSSI}{received signal strength indicator}
\acrodef{RSS}{received signal strength}
\acrodef{PDF}{probability distribution function}
\acrodef{RV}{random variable}
\acrodef{MSE}{mean square error}
\acrodef{HR}{hazard rate}
\acrodef{ITS}{intelligent transportation systems}
\acrodef{PSFCH}{physical sidelink feedback channel}
\acrodef{PUCCH}{physical uplink control channel}
\acrodef{RV}{random variable}
\acrodef{i.i.d.}{independent, identically distributed}
\acrodef{GMM}{Gaussian mixture model}
\acrodef{TVaR}{tail value at risk}
\acrodef{CRF}{conditional relative frequency}
\acrodef{ORF}{overall relative frequency}

\acrodef{ISAC}{integrated sensing and communication}
\acrodef{IOC}{inverse optimal control}
\acrodef{mmWave}{millimeter wave}
\acrodef{DoF}{degrees of freedom}
\acrodef{VRU}{vulnerable road user}
\acrodef{AV}{autonomous vehicle}
\acrodef{HD}{high-definition}
\acrodef{NCV}{non-connected vehicle}
\acrodef{CV}{connected vehicle}
\acrodef{AoA}{angle-of-arrival}
\acrodef{AoD}{angle-of-departure}
\acrodef{ULA}{uniform linear array}
\acrodef{EKF}{extended Kalman filtering}
\acrodef{SNR}{signal-to-noise ratio}
\acrodef{VU}{vehicular user}
\acrodef{RIS}{reconfigurable intelligent surface}
\acrodef{AWGN}{additive white Gaussian noise}
\acrodef{MLE}{maximum likelihood estimator}
\acrodef{RSU}{roadside unit}
\acrodef{RCS}{radar cross section}
\acrodef{AD}{autonomous driving}
\acrodef{CAV}{connected autonomous vehicle}
\acrodef{ECM}{electronic countermeasure}
\acrodef{CP}{cooperative perception}

%% file: Globecom_draft.bbl
\begin{thebibliography}{10}
\providecommand{\url}[1]{#1}
\csname url@samestyle\endcsname
\providecommand{\newblock}{\relax}
\providecommand{\bibinfo}[2]{#2}
\providecommand{\BIBentrySTDinterwordspacing}{\spaceskip=0pt\relax}
\providecommand{\BIBentryALTinterwordstretchfactor}{4}
\providecommand{\BIBentryALTinterwordspacing}{\spaceskip=\fontdimen2\font plus
\BIBentryALTinterwordstretchfactor\fontdimen3\font minus \fontdimen4\font\relax}
\providecommand{\BIBforeignlanguage}[2]{{%
\expandafter\ifx\csname l@#1\endcsname\relax
\typeout{** WARNING: IEEEtran.bst: No hyphenation pattern has been}%
\typeout{** loaded for the language `#1'. Using the pattern for}%
\typeout{** the default language instead.}%
\else
\language=\csname l@#1\endcsname
\fi
#2}}
\providecommand{\BIBdecl}{\relax}
\BIBdecl

\bibitem{9755276}
Z.~Wei, F.~Liu, C.~Masouros, N.~Su, and A.~P. Petropulu, ``{Toward Multi-Functional 6G Wireless Networks: Integrating Sensing, Communication, and Security},'' \emph{IEEE Communications Magazine}, vol.~60, no.~4, pp. 65--71, 2022.

\bibitem{9199556}
N.~Su, F.~Liu, and C.~Masouros, ``{Secure Radar-Communication Systems With Malicious Targets: Integrating Radar, Communications and Jamming Functionalities},'' \emph{IEEE Transactions on Wireless Communications}, vol.~20, no.~1, pp. 83--95, 2021.

\bibitem{9737364}
N.~Su, F.~Liu, Z.~Wei, Y.-F. Liu, and C.~Masouros, ``{Secure Dual-Functional Radar-Communication Transmission: Exploiting Interference for Resilience Against Target Eavesdropping},'' \emph{IEEE Transactions on Wireless Communications}, vol.~21, no.~9, pp. 7238--7252, 2022.

\bibitem{10227884}
N.~Su, F.~Liu, and C.~Masouros, ``{Sensing-Assisted Eavesdropper Estimation: An ISAC Breakthrough in Physical Layer Security},'' \emph{IEEE Transactions on Wireless Communications}, vol.~23, no.~4, pp. 3162--3174, 2024.

\bibitem{10781436}
Y.~Cao, L.~Duan, and R.~Zhang, ``{Sensing for Secure Communication in ISAC: Protocol Design and Beamforming Optimization},'' \emph{IEEE Transactions on Wireless Communications}, vol.~24, no.~2, pp. 1207--1220, 2025.

\bibitem{10839241}
B.~Zheng, X.~Xiong, T.~Ma, J.~Tang, D.~W.~K. Ng, A.~L. Swindlehurst, and R.~Zhang, ``{Intelligent Reflecting Surface-Enabled Anti-Detection for Secure Sensing and Communications},'' \emph{IEEE Wireless Communications}, vol.~32, no.~2, pp. 156--163, 2025.

\bibitem{10856886}
F.~Xu, W.~Lai, and K.~Shen, ``{Intelligent Surface Assisted Radar Stealth Against Unauthorized ISAC},'' \emph{IEEE Wireless Communications Letters}, vol.~14, no.~4, pp. 1149--1153, 2025.

\bibitem{10443321}
X.~Shao and R.~Zhang, ``{Target-Mounted Intelligent Reflecting Surface for Secure Wireless Sensing},'' \emph{IEEE Transactions on Wireless Communications}, vol.~23, no.~8, pp. 9745--9758, 2024.

\bibitem{10575930}
B.~Zheng, X.~Xiong, J.~Tang, and R.~Zhang, ``{Intelligent Reflecting Surface-Aided Electromagnetic Stealth Against Radar Detection},'' \emph{IEEE Transactions on Signal Processing}, vol.~72, pp. 3438--3452, 2024.

\bibitem{10634199}
H.~Wang, B.~Zheng, X.~Shao, and R.~Zhang, ``{Intelligent Reflecting Surface-Aided Radar Spoofing},'' \emph{IEEE Wireless Communications Letters}, vol.~13, no.~10, pp. 2722--2726, 2024.

\bibitem{9171304}
F.~Liu, W.~Yuan, C.~Masouros, and J.~Yuan, ``{Radar-Assisted Predictive Beamforming for Vehicular Links: Communication Served by Sensing},'' \emph{IEEE Transactions on Wireless Communications}, vol.~19, no.~11, pp. 7704--7719, 2020.

\bibitem{9732917}
G.~C. Trichopoulos, P.~Theofanopoulos, B.~Kashyap, A.~Shekhawat, A.~Modi, T.~Osman, S.~Kumar, A.~Sengar, A.~Chang, and A.~Alkhateeb, ``{Design and Evaluation of Reconfigurable Intelligent Surfaces in Real-World Environment},'' \emph{IEEE Open Journal of the Communications Society}, vol.~3, pp. 462--474, 2022.

\bibitem{9947033}
Z.~Du, F.~Liu, W.~Yuan, C.~Masouros, Z.~Zhang, S.~Xia, and G.~Caire, ``{Integrated Sensing and Communications for V2I Networks: Dynamic Predictive Beamforming for Extended Vehicle Targets},'' \emph{IEEE Transactions on Wireless Communications}, vol.~22, no.~6, pp. 3612--3627, 2023.

\bibitem{10423078}
J.~An, C.~Xu, D.~W.~K. Ng, C.~Yuen, and L.~Hanzo, ``{Adjustable-Delay RIS Is Capable of Improving OFDM Systems},'' \emph{IEEE Transactions on Vehicular Technology}, vol.~73, no.~7, pp. 9927--9942, 2024.

\bibitem{9724202}
X.~Shao, C.~You, W.~Ma, X.~Chen, and R.~Zhang, ``{Target Sensing With Intelligent Reflecting Surface: Architecture and Performance},'' \emph{IEEE Journal on Selected Areas in Communications}, vol.~40, no.~7, pp. 2070--2084, 2022.

\bibitem{10740590}
R.~Li, X.~Shao, S.~Sun, M.~Tao, and R.~Zhang, ``{IRS Aided Millimeter-Wave Sensing and Communication: Beam Scanning, Beam Splitting, and Performance Analysis},'' \emph{IEEE Transactions on Wireless Communications}, vol.~23, no.~12, pp. 19\,713--19\,727, 2024.

\end{thebibliography}
